%% file: stabilityofskineffect2.tex
\documentclass[a4paper,reqno]{amsart}

\usepackage[british]{babel}

\numberwithin{equation}{section}

\usepackage{amssymb}
\usepackage[protrusion=true,expansion=true]{microtype}	
\usepackage{amsmath,amsfonts,amsthm} 
\usepackage[pdftex]{graphicx}	
\usepackage{url}
\usepackage[title,titletoc,page]{appendix} 
\usepackage{listings} 
\usepackage{color}
\usepackage{comment}
\usepackage{caption}
\usepackage{subcaption}
\usepackage{algorithm}

\usepackage{fancyhdr}
\numberwithin{equation}{section}		
\numberwithin{figure}{section}			
\numberwithin{table}{section}				

\newcommand{\babs}[1]{\left|{#1}\right|}
\newcommand{\bnorm}[1]{\left|\left|{#1}\right|\right|}

\newcommand{\vect}[1]{\boldsymbol{\mathbf{#1}}}

\usepackage[left=3.5cm,right=3.5cm,top=3.5cm,bottom=3.5cm,footskip=1.5cm,headsep=1cm]{geometry}
\usepackage[%
pdfauthor={Habib Ammari, Silvio Barandun, Bryn Davies, Erik Orvehed Hiltunen and Ping Liu},%
pdftitle={Stability of the skin effect},
pdfsubject={Stability of the skin effect},
pdfkeywords={Non-Hermitian systems, skin effect, condensation, subwavelength resonators, imaginary gauge
potential, phase transition, localisation},
]{hyperref}

\input{comands}

\begin{document}
\title[Stability of the non-Hermitian skin effect]{Stability of the non-Hermitian skin effect}

 \author[H. Ammari]{Habib Ammari}
\address{\parbox{\linewidth}{Habib Ammari\\
 ETH Z\"urich, Department of Mathematics, Rämistrasse 101, 8092 Z\"urich, Switzerland.}}
\email{habib.ammari@math.ethz.ch}
\thanks{}

\author[S. Barandun]{Silvio Barandun}
 \address{\parbox{\linewidth}{Silvio Barandun\\
 ETH Z\"urich, Department of Mathematics, Rämistrasse 101, 8092 Z\"urich, Switzerland.}}
 \email{silvio.barandun@sam.math.ethz.ch}

\author[B. Davies]{Bryn Davies}
 \address{\parbox{\linewidth}{Bryn Davies\\
Department of Mathematics, Imperial College London, 180 Queen's Gate, London SW7~2AZ, UK.}}
\email{bryn.davies@imperial.ac.uk}

 \author[E.O. Hiltunen]{Erik Orvehed Hiltunen}
\address{\parbox{\linewidth}{Erik Orvehed Hiltunen\\
Department of Mathematics, Yale University, 10 Hillhouse Ave,
New Haven, CT~06511, USA.}}
\email{erik.hiltunen@yale.edu}

\author[P. Liu]{Ping Liu}
 \address{\parbox{\linewidth}{Ping Liu\\
 ETH Z\"urich, Department of Mathematics, Rämistrasse 101, 8092 Z\"urich, Switzerland.}}
\email{ping.liu@sam.math.ethz.ch}

\maketitle

\begin{abstract}
This paper shows that the skin effect in systems of non-Hermitian subwavelength resonators is robust with respect to random imperfections in the system. The subwavelength resonators are highly contrasting material inclusions that resonate in a low-frequency regime. The non-Hermiticity is due to the introduction of an imaginary gauge potential, which leads to a skin effect that is manifested by the system's eigenmodes accumulating at one edge of the structure. We elucidate the topological protection of the associated (real) eigenfrequencies and illustrate the competition between the two different localisation effects present when the system is randomly perturbed: the non-Hermitian skin effect and the disorder-induced Anderson localisation. We show that, as the strength of the disorder increases, more and more eigenmodes become localised in the bulk. Our results are based on an asymptotic matrix model for subwavelength physics and can be generalised also to tight-binding models in condensed matter theory. 

\end{abstract}

\date{}

\bigskip

\par

\bigskip

\\

\section{Introduction}

The skin effect is the phenomenon whereby a large proportion of the bulk eigenmodes of a non-Hermitian system are localised at one edge of an open chain \cite{okuma,yuto}. In subwavelength physics, it emerges in an array of subwavelength resonators when an imaginary gauge potential is introduced inside the resonators, which are much smaller than the operating wavelength. The resonance of these subwavelength structures (whose dimensions are substantially smaller than the operating wavelength) is essential as without exciting the structure's subwavelength resonances the effect of the imaginary gauge potential would be negligible.

In systems of subwavelength resonators, the skin effect phenomenon is unique to non-Hermitian systems with non-reciprocal coupling. While localisation of specific eigenmodes can be achieved in other (\emph{e.g.} Hermitian) systems, the skin effect is characterised by having a large number of the modes (which scales with the size of the system) localised at one edge of the system. This phenomenon has been realised experimentally in both photonic and phononic systems \cite{ghatak.brandenbourger.ea2020Observation,skinadd1,skinadd2, skinadd3,kai,kawabata}. It significantly advances the field of active metamaterials and opens new avenues to channel and manipulate energy at subwavelength scales \cite{applications}.

The non-Hermitian skin effect was first introduced in condensed matter physics as a non-Hermitian extension of the Anderson model of localisation \cite{hatano.nelson1996Localization}. The imaginary gauge potential leads to the simultaneous condensation of an extensive number of bulk eigenmodes, all in the same direction \cite{hatano,yokomizo.yoda.ea2022NonHermitian, rivero.feng.ea2022Imaginary,longhi}. The tight binding models used in condensed matter theory share many fundamental similarities with the one-dimensional subwavelength classical wave system considered here. 

In a recent work \cite{ammari2023mathematical}, the non-Hermitian skin effect in the subwavelength regime was studied using first-principle mathematical analysis. One-dimensional systems of subwavelength resonators were considered, with an imaginary gauge potential added to break Hermiticity. Explicit asymptotic expressions for the subwavelength eigenfrequencies and eigenmodes were obtained using a gauge capacitance matrix formulation of the problem (which is a reformulation of the standard capacitance matrices that are commonplace in Hermitian subwavelength physics and electrostatics). Moreover, the exponential decay of eigenmodes and their accumulation at one edge of the structure (the non-Hermitian skin effect) was shown to be induced by the Fredholm index of an associated Toeplitz operator. A remaining open question is whether the skin effect is stable with respect to disorder. This important problem has been subject to recent debate in the physics and engineering communities \cite{okuma}. 

In this paper, we prove the robustness of the non-Hermitian skin effect with respect to random imperfections in the system. Based on delicate eigenvalue and eigenvector analysis of perturbed ``almost-Toeplitz'' matrices, we quantify the stability of the non-Hermitian skin effect. Moreover, we illustrate the competition between the non-Hermitian skin effect and Anderson localisation. Anderson localisation here refers to strong localisation of eigenmodes in the bulk (at subwavelength scales) that it is induced by disorder \cite{anderson1958absence}. We observe that, as the disorder strength increases, more and more eigenmodes are  localised in the bulk. This leads to a disorder-induced phase transition (in terms of the disorder strength) between  accumulation at one edge of the structure and localisation in the bulk. As far as we know, these findings provide the first justification to the experimental results discussed in \cite{andersoninterplay1, andersoninterplay2,andersoninterplay3}. It also extends the Anderson localisation in systems of subwavelength resonators \cite{ouranderson2022} to the non-Hermitian case. On the other hand, we also elucidate the topological protection of the (real) eigenfrequencies associated with the eigenmodes that accumulate at one edge of the structure. All of these eigenfrequencies stay inside a region of the complex plane with nontrivial winding number and can, consequently, be said to be topologically protected. Conversely, the eigenfrequencies corresponding to eigenmodes that are localised in the bulk fall outside of this region. 

The paper is organised as follows. In Section \ref{sect1}, we present the mathematical setup of the problem and recall its discrete formulation which provides approximations of the eigenfrequencies and eigenmodes of a finite chain of subwavelength resonators in terms of the eigenvalues and eigenvectors of the gauge capacitance matrix. Given this discrete formulation and the effect of uncertainties in the positions of the resonators or their material parameters, we can reduce the stability analysis to the analysis of a perturbed almost-Toeplitz matrices. Section \ref{sect2} is devoted to the stability analysis of the eigenvalues while in Section \ref{sect3} we prove the stability of the eigenvectors and show their exponential decay and condensation at one edge of the structure. In Section \ref{sect4}, we numerically illustrate our main findings in this paper. Moreover, we  show how condensation of the eigenmodes at the edge and localisation in the bulk are competing effects and present topologically-induced phase transition diagrams in terms of the strength of the disorder. We show numerically that as the strength of the disorder increases, the number of eigenmodes localised in the bulk increases. We also elucidate the fact that the non-trivial winding of the symbol of the associated Toeplitz operator at the eigenfrequencies protects an extensive number of associated eigenmodes from localisation. The paper ends with some concluding remarks and interesting generalisations of the results. 

\section{Non-Hermitian skin effect} \label{sect1}
We begin this section by introducing the setting and recalling results from \cite{ammari2023mathematical} on the non-Hermitian skin effect without disorder. In \Cref{sec:rand}, we introduce the disordered model which will be studied in subsequent sections.

\subsection{Problem formulation}
We consider a one-dimensional chain of $N$ disjoint identical subwavelength resonators $D_i\coloneqq (x_i^{\iL},x_i^{\iR})$, where $(x_i^{\iLR})_{1\<i\<N} \subset \R$ are the $2N$ extremities satisfying $x_i^{\iL} < x_i^{\iR} <  x_{i+1}^{\iL}$ for any $1\leq i \leq N$. We fix the coordinates such that $x_1^{\iL}=0$. We also denote by  $\ell_i = x_i^{\iR} - x_i^{\iL}$ the length of each of the resonators,  and by $s_i= x_{i+1}^{\iL} -x_i^{\iR}$ the spacing between the $i$-th and $(i+1)$-th resonators. The system is illustrated in \cref{fig:setting}. We use 
\begin{align*}
   D\coloneqq \bigcup_{i=1}^N(x_i^{\iL},x_i^{\iR})
\end{align*}
to symbolise the set of subwavelength resonators. In this paper, we only consider systems of equally spaced identical resonators, that is,
\begin{align*}
    \ell_i = \ell \in \R_{>0}\text{ for all } 1\leq i\leq N \quad \text{and} \quad s_i = s \in \R_{>0}  \text{ for all } 1\leq i\leq N-1.
\end{align*}
This will simplify the formulas in our subsequent analysis and is sufficient to understand the fundamental mechanisms  behind the skin and localisation effects we are interested in.

\begin{figure}[htb]
    \centering
    \begin{adjustbox}{width=\textwidth}
    \begin{tikzpicture}
        \coordinate (x1l) at (1,0);
        \path (x1l) +(1,0) coordinate (x1r);
        \path (x1r) +(0.5,0.7) coordinate (s1);
        \path (x1r) +(1,0) coordinate (x2l);
        \path (x2l) +(1,0) coordinate (x2r);
        \path (x2r) +(0.5,0.7) coordinate (s2);
        \path (x2r) +(1,0) coordinate (x3l);
        \path (x3l) +(1,0) coordinate (x3r);
        \path (x3r) +(.5,0.7) coordinate (s3);
        \path (x3r) +(1,0) coordinate (x4l);
        \path (x4l) +(1,0) coordinate (x4r);
        \path (x4r) +(0.5,0.7) coordinate (s4);
        \path (x4r) +(1,0) coordinate (dots);
        \path (dots) +(1,0) coordinate (x5l);
        \path (x5l) +(1,0) coordinate (x5r);
        \path (x5r) +(1.,0) coordinate (x6l);
        \path (x5r) +(0.5,0.7) coordinate (s5);
        \path (x6l) +(1,0) coordinate (x6r);
        \path (x6r) +(1.,0) coordinate (x7l);
        \path (x6r) +(0.5,0.7) coordinate (s6);
        \path (x7l) +(1,0) coordinate (x7r);
        \path (x7r) +(1.,0) coordinate (x8l);
        \path (x7r) +(0.5,0.7) coordinate (s7);
        \path (x8l) +(1,0) coordinate (x8r);
        \draw[ultra thick] (x1l) -- (x1r);
        \node[anchor=north] (label1) at (x1l) {$x_1^{\iL}$};
        \node[anchor=north] (label1) at (x1r) {$x_1^{\iR}$};
        \node[anchor=south] (label1) at ($(x1l)!0.5!(x1r)$) {$\ell$};
        \draw[dotted,|-|] ($(x1r)+(0,0.25)$) -- ($(x2l)+(0,0.25)$);
        \draw[ultra thick] (x2l) -- (x2r);
        \node[anchor=north] (label1) at (x2l) {$x_2^{\iL}$};
        \node[anchor=north] (label1) at (x2r) {$x_2^{\iR}$};
        \node[anchor=south] (label1) at ($(x2l)!0.5!(x2r)$) {$\ell$};
        \draw[dotted,|-|] ($(x2r)+(0,0.25)$) -- ($(x3l)+(0,0.25)$);
        \draw[ultra thick] (x3l) -- (x3r);
        \node[anchor=north] (label1) at (x3l) {$x_3^{\iL}$};
        \node[anchor=north] (label1) at (x3r) {$x_3^{\iR}$};
        \node[anchor=south] (label1) at ($(x3l)!0.5!(x3r)$) {$\ell$};
        \draw[dotted,|-|] ($(x3r)+(0,0.25)$) -- ($(x4l)+(0,0.25)$);
        \node (dots) at (dots) {\dots};
        \draw[ultra thick] (x4l) -- (x4r);
        \node[anchor=north] (label1) at (x4l) {$x_4^{\iL}$};
        \node[anchor=north] (label1) at (x4r) {$x_4^{\iR}$};
        \node[anchor=south] (label1) at ($(x4l)!0.5!(x4r)$) {$\ell$};
        \draw[dotted,|-|] ($(x4r)+(0,0.25)$) -- ($(dots)+(-.25,0.25)$);
        \draw[ultra thick] (x5l) -- (x5r);
        \node[anchor=north] (label1) at (x5l) {$x_{N-3}^{\iL}$};
        \node[anchor=north] (label1) at (x5r) {$x_{N-3}^{\iR}$};
        \node[anchor=south] (label1) at ($(x5l)!0.5!(x5r)$) {$\ell$};
        \draw[dotted,|-|] ($(x5r)+(0,0.25)$) -- ($(x6l)+(0,0.25)$);
        \draw[ultra thick] (x6l) -- (x6r);
        \node[anchor=north] (label1) at (x6l) {$x_{N-2}^{\iL}$};
        \node[anchor=north] (label1) at (x6r) {$x_{N-2}^{\iR}$};
        \node[anchor=south] (label1) at ($(x6l)!0.5!(x6r)$) {$\ell$};
        \draw[dotted,|-|] ($(x6r)+(0,0.25)$) -- ($(x7l)+(0,0.25)$);
        \draw[ultra thick] (x7l) -- (x7r);
        \node[anchor=north] (label1) at (x7l) {$x_{N-1}^{\iL}$};
        \node[anchor=north] (label1) at (x7r) {$x_{N-1}^{\iR}$};
        \node[anchor=south] (label1) at ($(x7l)!0.5!(x7r)$) {$\ell$};
        \draw[dotted,|-|] ($(x7r)+(0,0.25)$) -- ($(x8l)+(0,0.25)$);
        \draw[ultra thick] (x8l) -- (x8r);
        \node[anchor=north] (label1) at (x8l) {$x_{N}^{\iL}$};
        \node[anchor=north] (label1) at (x8r) {$x_{N}^{\iR}$};
        \node[anchor=south] (label1) at ($(x8l)!0.5!(x8r)$) {$\ell$};
        \node[anchor=north] (label1) at (s1) {$s$};
        \node[anchor=north] (label1) at (s2) {$s$};
        \node[anchor=north] (label1) at (s3) {$s$};
        \node[anchor=north] (label1) at (s4) {$s$};
        \node[anchor=north] (label1) at (s5) {$s$};
        \node[anchor=north] (label1) at (s6) {$s$};
        \node[anchor=north] (label1) at (s7) {$s$};
    \end{tikzpicture}
    \end{adjustbox}
    \caption{A chain of $N$ one-dimensional subwavelength resonators, with length $\ell$ and spacing $s$.}
    \label{fig:setting}
\end{figure}
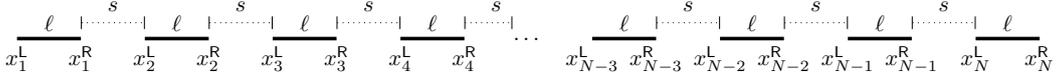

In this work, we consider the following one-dimensional 
damped wave equation where the damping acts in the space dimension instead of the time dimension:
\begin{align}
    -\frac{\omega^{2}}{\kappa(x)}u(x)- \gamma(x) \frac{\dd}{\dd x}u(x)-\frac{\dd}{\dd x}\left( \frac{1}{\rho(x)}\frac{\dd}{\dd
    x}  u(x)\right) =0,\qquad x \in\R,
    \label{eq: gen Strum-Liouville}
\end{align}
for a piecewise constant damping coefficient
\begin{align}\label{equ:nonhermitiancoeff1}
\gamma(x) = \begin{dcases}
    \gamma,\quad x\in D,\\
    0,\quad x \in \R\setminus D.
\end{dcases}
\end{align}
The parameter $\gamma$ extends the usual scalar wave equation to a generalised Strum--Liouville equation via the introduction of an imaginary gauge potential \cite{yokomizo.yoda.ea2022NonHermitiana}. 
The material parameters $\kappa(x)$ and $\rho(x)$ are  piecewise constant
\begin{align*}
    \kappa(x)=
    \begin{dcases}
        \kappa_b & x\in D,\\
        \kappa&  x\in\R\setminus D,
    \end{dcases}\quad\text{and}\quad
    \rho(x)=
    \begin{dcases}
        \rho_b & x\in D,\\
        \rho&  x\in\R\setminus D,
    \end{dcases}
\end{align*}
where the constants $\rho_b, \rho, \kappa, \kappa_b \in \R_{>0}$. The wave speeds inside the resonators $D$ and inside the background medium $\R\setminus D$, are denoted respectively by $v_b$ and $v$, 
the wave numbers respectively by $k_b$ and $k$, the frequency by $\omega$, and the contrast between the densities of the resonators and the background medium by $\delta$:
\begin{align}
    v_b:=\sqrt{\frac{\kappa_b}{\rho_b}}, \qquad v:=\sqrt{\frac{\kappa}{\rho}},\qquad
    k_b:=\frac{\omega}{v_b},\qquad k:=\frac{\omega}{v},\qquad
    \delta:=\frac{\rho_b}{\rho}.
\end{align}

We are interested in the resonances $\omega\in\C$ such that \eqref{eq: gen Strum-Liouville} has a non-trivial solution in a high-contrast, low-frequency (subwavelength) regime. This regime is typically characterised by letting the contrast parameter $\delta\to 0$ and looking for solutions which are such that $\omega \to 0$ as $\delta\to 0$. One consequence of this asymptotic ansatz is that it lends itself to characterisation using asymptotic analysis \cite{ammari.davies.ea2021Functional}. Note that this limit recovers subwavelength resonances, while keeping the size of the resonators fixed. 

In \cite{ammari2023mathematical}, an asymptotic analysis in the subwavelength limit was performed on the system of non-Hermitian one-dimensional subwavelength resonators considered here. It was shown that the resonances are given by the eigenstates of the \emph{gauge capacitance matrix} $\capmat^\gamma$. This is a modified version of the conventional capacitance matrix that is often used to characterise many-body low-frequency resonance problems; see, for instance, \cite{ammari.davies.ea2021Functional}.

The following results are from \cite{ammari2023mathematical}.
\begin{theorem} \label{thm:gauge}
Let the gauge capacitance matrix $\capmat^\gamma =(\capmat_{i,j}^\gamma)^N_{i,j=1}$ be defined by
\begin{align}
        \label{eq: cap mat ESI}
            \capmat_{i,j}^\gamma \coloneqq \begin{dcases}
                \frac{\gamma}{s} \frac{1}{1-e^{-\gamma\ell}}, & i=j=1,\\
                \frac{\gamma}{s} \coth(\gamma\ell/2), & 1< i=j< N,\\
                \pm\frac{\gamma}{ s} \frac{1}{1-e^{\pm\gamma}}, & 1\leq i=j \pm 1\leq N,\\
                -\frac{\gamma}{ s} \frac{1}{1-e^{\gamma\ell}}, & i=j=N.\\
            \end{dcases}
            \end{align} 
            Then, 
            \begin{enumerate}
            \item[(i)] All the eigenvalues of $\capmat^\gamma$ are real. They are given by
\begin{align}
    \lambda_1 &= 0,\nonumber\\
    \lambda_k &= \frac{\gamma}{s} \coth(\gamma\ell/2)+\frac{2\abs{\gamma}}{s}\frac{e^{\frac{\gamma\ell}{2}}}{\vert e^{\gamma\ell}-1\vert}\cos\left(\frac{\pi}{N}k\right), \quad 2\leq k\leq N . \label{eq: eigenvalues capmat}
\end{align}
Furthermore, the associated eigenvectors $\bm a_k$ satisfy the following inequality, for $2\leq k\leq N$
\begin{align}
\vert \bm a_k^{(i)}\vert \leq \kappa_k e^{-\gamma\ell\frac{i-1}{2}}\quad \text{for all } 1\leq i\leq N \label{eq: decay for eigemodes},
\end{align}
for some $\kappa_k\leq (1+e^{\frac{\gamma\ell}{2}})^2$. Here, 
 $\bm a_k^{(i)}$ denotes the $i$-th entry of the eigenvector  $\bm a_k$;
 \item[(ii)] The  $N$ subwavelength eigenfrequencies $\omega_i$ of (\ref{eq: gen Strum-Liouville}) satisfy, as $\delta\to0$,
    \begin{align*}
        \omega_i =  v_b \sqrt{\delta\lambda_i} + \BO(\delta),
    \end{align*}
    where $(\lambda_i)_{1\leq i\leq N}$ are the eigenvalues of 
   $\capmat^\gamma$. Furthermore, let $u_i$ be a subwavelength eigenmode corresponding to $\omega_i$ and let $\bm a_i$ be the corresponding eigenvector of $\capmatg$. Then
        \begin{align*}
            u_i(x) = \sum_j \bm a_i^{(j)}V_j(x) + \BO(\delta),
        \end{align*}
        where $V_j$ are defined by 
        \begin{align}
    \begin{dcases}
        -\frac{\dd{^2}}{\dd x^2} V_i =0, & x\in\R\setminus\bigcup_{i=1}^N(x_i^{\iL},x_i^{\iR}), \\
        V_i(x)=\delta_{ij}, & x\in (x_j^{\iL},x_j^{\iR}),\\
        V_i(x) = \BO(1) & \mbox{as } \abs{x}\to\infty.
    \end{dcases}
    \label{eq: def V_i}
\end{align}    
    
\end{enumerate}  
\end{theorem}
From Theorem \ref{thm:gauge}, we can see that $\capmatg$
is \emph{almost} Toeplitz (in the sense that it has constant diagonals other than deterministic perturbations in the corners) and its eigenvectors display exponential decay both with respect to the site index $i$ and the factor $\gamma$. This shows the condensation of bulk eigenmodes at one of the edges of the system of subwavelength resonators. The exponential decay of the eigenvectors is directly linked with a topological property. Let $T$ be a Toeplitz operator with continuous symbol $f_T$ and $T_N$ be the truncation of $T$ to the upper-left $N\times N$ submatrix. Then, if $f_T$ is sufficiently smooth, standard Toeplitz theory says that any eigenvalue $\lambda\in\C$ which is such that the winding $w(f_T,\lambda)$ of $f_T$ at $\lambda$ is negative will be such that the corresponding eigenvector decays exponentially. We show this topological region in \cref{fig: winding region} for our system. In \cref{fig: condendensed eigenvectors unperturbed} we plot the eigenvectors superimposed on one another to portray condensation on the left edge of the structure; $\lambda_1=0$ corresponds to a trivial (constant) eigenvector while all other eigenvectors are exponentially localised to the left edge of the structure.

\begin{figure}[h]
    \begin{subfigure}[t]{0.45\textwidth}
        \centering
    \includegraphics[height=0.75\textwidth]{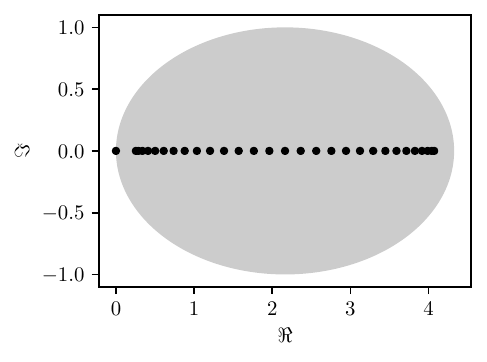}
    \caption{Shaded in grey is the region where the symbol $f_T$ has negative winding, for $T$ the Toeplitz operator corresponding to the capacitance matrix $\mathcal{C}^\gamma$. The black dots show the spectrum of $\mathcal{C}^\gamma$.}\label{fig: winding region}
    \end{subfigure}\hfill
    \begin{subfigure}[t]{0.45\textwidth}
        \centering
    \includegraphics[height=0.75\textwidth]{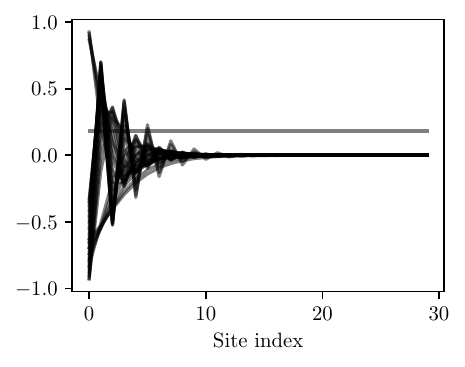}
    \caption{Eigenvector condensation on the left edge of the structure. The 30 modes are plotted, superimposed. The constant eigenvector is the one associated to the eigenvalue $0$.}\label{fig: condendensed eigenvectors unperturbed}
    \end{subfigure}
    \caption{Numerical simulations for an unperturbed system of $N = 30$ resonators with $s = \ell = 1$ and $\gamma=1$.}
\end{figure}

\subsection{Randomly perturbed gauge capacitance matrix} \label{sec:rand}
To simplify the notation, we denote the tridiagonal Toeplitz matrix with deterministic perturbations on diagonal corners by 
\begin{equation}\label{equ:toeplitzcornerperturbed1}
    T_N^{(a, b)}=\left(\begin{array}{ccccccc}
        \alpha+a & \beta & 0 & 0 & \ldots & 0 & 0\\
        \eta & \alpha & \beta & 0 & \ldots & 0 & 0 \\
        0 & \eta & \alpha & \beta & \ldots & 0 & 0 \\
        \ldots & \ldots & \ldots & \ldots & \ldots & \ldots & \ldots \\
        \ldots & \ldots & \ldots & \ldots & \ldots & \alpha & \beta \\
        0 & 0 & 0 & 0 & \ldots & \eta & \alpha+b
    \end{array}\right).
\end{equation}
Throughout the paper, we only consider $T_N^{(a,b)}$ with $\alpha, \eta, \beta, a, b\in \mathbb R$ and $\eta\beta>0$. In particular, observe from (\ref{eq: cap mat ESI}) that $$\capmatg = T_N^{(\eta, \beta)}$$ and $\eta, \alpha, \beta$ are such that $$a=\eta,\quad b=\beta,\quad \eta \beta > 0 \quad \text{and,} \quad  \eta+\alpha+\beta=0.$$ Note also that $T_{N}^{(0,0)}$ is a tridiagonal Toeplitz matrix.

In order to study the stability of the non-Hermitian skin effect with respect to random imperfections in the system design, we either add random errors to the positions of the resonators (keeping the length of the resonators unchanged) or to the $\gamma$-term and then repeatedly compute the subwavelength eigenfrequencies and eigenmodes. The perturbations in the positions and in the values of the $\gamma$-parameter are drawn at random from uniform distributions with zero-mean values. 
Since these random perturbations affect only the tridiagonal entries of the gauge capacitance matrix $\hat{\capmat}^\gamma$ of the randomly perturbed system, we can write in both cases that
\begin{equation}\label{equ:toeplitzallperturbed1}
\hat{\capmat}^\gamma = \widehat T_N^{(a,b)}:=\left(\begin{array}{cccccc}
    \alpha + a +\epsilon_{\alpha, 1} & \beta+\epsilon_{\beta, 1}  & 0  & \ldots & 0 & 0\\
    \eta +\epsilon_{\eta, 2}  & \alpha + \epsilon_{\alpha, 2} & \beta +\epsilon_{\beta, 2}  & \ldots & 0 & 0 \\
    0 & \eta+\epsilon_{\eta, 3} & \alpha+\epsilon_{\alpha, 3}  & \ldots & 0 & 0 \\
    \ldots & \ldots & \ldots & \ldots & \ldots  & \ldots \\
    \ldots & \ldots & \ldots  & \ldots & \alpha+\epsilon_{\alpha, N-1} & \beta+\epsilon_{\beta, N-1} \\
    0 & 0 & 0 & \ldots & \eta+\epsilon_{\eta, N} & \alpha+b+\epsilon_{\alpha,N}
\end{array}\right)
\end{equation}
with $a=\eta, b=\beta$. 

\section{Stability of eigenvalues} \label{sect2}
In this section, we derive stability results for the eigenvalues of $T_N^{(a,b)}$. This relies on a crucial observation that the tridiagonal matrix $T_N^{(a,b)}$ always has same eigenvalues as a Hermitian matrix. Thus we first recall the following well known Weyl theorem for the stability of eigenvalues of Hermitian matrices; see \cite[Theorem 1.1.7]{reis2011matrix}, \cite[Theorem 8.1.6]{golub2013matrix} and \cite[Theorem 10.3.1]{parlett1998symmetric}. We also refer the reader to \cite{ipsen2009refined} for some more refined results. 

\begin{theorem}\label{thm:weyltheorem1}
Let $A$ and $E$ be $N \times N$ Hermitian matrices. For $k \in\{1, \ldots, N\}$ denote by $\lambda_k(A+E), \lambda_k(A)$ the $k$-th eigenvalue of $A+E$ and $A$, respectively. Assume these to be arranged in a decreasing sequence. Then 
\[
\babs{\lambda_k(A+E)-\lambda_k(A)}\leq \Vert{E}\Vert_2.
\]
\end{theorem}

We next recall the following result \cite[Lemma A.6]{ammari2023mathematical} on the eigenvalues of $T_N^{(\eta, \beta)}$ where its proof comes from \cite{yueh2008explicit}. 

\begin{lemma}\label{lem:eigenvalues0}
Suppose that $\eta+\alpha+\beta=0$. Let $\lambda$ be an eigenvalue of $T_N^{(\eta, \beta)}$. Then, either $\lambda=\lambda_1:=0$ and the corresponding eigenvector is $\vect x_1=\mathbf{1}$ or
	\begin{equation}\label{equ:eigenvalueequ1}
	\lambda_k:=\alpha+2 \sqrt{\eta \beta } \cos \left(\frac{\pi}{N}(k-1)\right), \quad 2 \leq k \leq N,
	\end{equation}
	and the corresponding eigenvector is $\vect x_k$, 
 with entries
	\begin{equation}\label{equ:eigenvectorequ1}
	\vect x_k^{(j)}=\left(\frac{\eta}{\beta}\right)^{\frac{j-1}{2}}\left(\eta \sin \left(\frac{j(k-1) \pi}{N}\right)-\eta \sqrt{\frac{\eta}{\beta}} \sin \left(\frac{(j-1)(k-1) \pi}{N}\right)\right),\ j=1,\cdots, N.
	\end{equation}
\end{lemma}

\medskip
Then we state our result on the stability of the eigenvalues of $T_{N}^{(a,b)}$ with $a, b\in \mathbb R$.
\begin{theorem}\label{thm:eigenvaluestability0}
The eigenvalues of $T_N^{(a,b)}$ and $\widehat T_N^{(a,b)}$ are all real numbers. Let $\{\lambda_k\}, \{\widehat \lambda_k \}$ be respectively the eigenvalues of $T_N^{(a,b)}$ and $\widehat T_N^{(a,b)}$, arranged in  decreasing sequences. Assuming that
\begin{equation}
\max_{j=1,\dots, N}\left(\babs{\epsilon_{\eta, j}}, \babs{\epsilon_{\alpha, j}}, \babs{\epsilon_{\beta, j}}\right)\eqqcolon \epsilon,
\end{equation}
then we have
\[
\babs{\widehat \lambda_k-\lambda_k}<C_1(\eta, \beta, \epsilon)\epsilon,
\]
where 
\begin{equation}\label{equ:eigenvaluestability2}
C_1(\eta, \beta, \epsilon)= \frac{\babs{\beta} +\babs{\eta} +\epsilon}{\sqrt{\beta \eta}}+1.
\end{equation}
\end{theorem}
\begin{proof}
Note that all the $\lambda_k$'s are real since $T_N^{(a, b)}$ has the same eigenvalues as the Hermitian matrix
\begin{equation}\label{equ:proofeigenvaluestability0}
A= \left(\begin{array}{ccccccc}
	\alpha+a & \sqrt{\eta\beta} & 0 & 0 & \ldots & 0 & 0\\
	\sqrt{\eta\beta} & \alpha &  \sqrt{\eta\beta} & 0 & \ldots & 0 & 0 \\
	0 &  \sqrt{\eta\beta}  & \alpha &  \sqrt{\eta\beta} & \ldots & 0 & 0 \\
	\ldots & \ldots & \ldots & \ldots & \ldots & \ldots & \ldots \\
	\ldots & \ldots & \ldots & \ldots & \ldots & \alpha &  \sqrt{\eta\beta}  \\
	0 & 0 & 0 & 0 & \ldots &  \sqrt{\eta\beta}  & \alpha+b
\end{array}\right)
\end{equation}
 which can by seen from 
 $\babs{xI-T_N^{(a,b)}} = \babs{xI-A}$ by expanding the determinant along the last row. Here, $\babs{\cdot}$ denotes the determinant. 
 In the same manner, all the $\widehat \lambda_k$'s are real, as $\widehat T_N^{(a, b)}$ has the same eigenvalues as the Hermitian matrix
\[
\tilde{T}_N^{(a,b)}=\left(\begin{array}{ccccccc}
	\alpha+a+\epsilon_{\alpha, 1} & x_{1, 2}& 0 & 0 & \ldots & 0 & 0\\
x_{2,1}  & \alpha+ \epsilon_{\alpha, 2} & x_{2,3}& 0 & \ldots & 0 & 0 \\
	0 & x_{3,2}& \alpha+\epsilon_{\alpha, 3} & x_{3,4} & \ldots & 0 & 0 \\
	\ldots & \ldots & \ldots & \ldots & \ldots & \ldots & \ldots \\
	\ldots & \ldots & \ldots & \ldots & \ldots & \alpha+\epsilon_{\alpha, N-1} & x_{N-1, N} \\
	0 & 0 & 0 & 0 & \ldots &x_{N, N-1} & \alpha+b+\epsilon_{\alpha,n}
\end{array}\right),
\] 	
where 
\[
x_{j, j+1} = x_{j+1, j} =\sqrt{(\beta+\epsilon_{\beta,j})(\eta+\epsilon_{\eta, j+1})}, \quad j=1, \cdots, N-1.
\]
Now, we can make use of Theorem \ref{thm:weyltheorem1} to analyse the stability of the eigenvalues of $A$. Let 
\[
\tilde{\epsilon}_{j, j+1} = \tilde{\epsilon}_{j+1, j} =  \sqrt{(\beta+\epsilon_{\beta,j})(\eta+\epsilon_{\eta, j+1})}-\sqrt{\beta\eta}, \quad j=1, \cdots, N-1,
\]
and 
\[
\tilde{\epsilon}_{j, j} = \epsilon_{\alpha, j}, \quad j=1,\cdots, N.
\]
It is not hard to see that 
\[
\babs{\tilde{\epsilon}_{j, j+1} }< \left(\frac{\babs{\beta}+\babs{\eta}+\epsilon}{2\sqrt{\beta \eta}}\right)\epsilon, \quad j=1,\cdots, N-1,
\]
as $\epsilon\leq 1$. We decompose $\tilde{T}_N^{(a,b)}$ as
\[
A+E,
\]
where $E_{i, j}= \tilde{\epsilon}_{i, j}$, for $i=j, i=j+1, i= j-1$, and $E_{i,j}=0$ for other $i, j$. In particular, the following estimate holds:
\begin{align*}
\bnorm{E}_2 =& \max_{\bnorm{v}_2=1}\bnorm{Ev}_2 < \left(\frac{\babs{\beta}+\babs{\eta}+\epsilon}{\sqrt{\beta \eta}}+1\right)\epsilon.
\end{align*}
The above estimate can be derived from the fact that
\begin{align*}
Ev=&\begin{pmatrix}
	\tilde \epsilon_{1,1}v_1 + \tilde \epsilon_{1,2}v_2\\
	\tilde \epsilon_{2,1}v_1 + \tilde \epsilon_{2,2}v_2+ \tilde\epsilon_{2, 3}v_3\\
	\vdots \\
	\tilde \epsilon_{N-1,N-2}v_{N-2} + \tilde \epsilon_{N-1,N-1}v_{N-1}+ \tilde\epsilon_{N-1, N}v_N\\
	\tilde \epsilon_{N,N-1}v_{N-1} + \tilde \epsilon_{N,N}v_N
\end{pmatrix} \\
=&\begin{pmatrix}
	0\\
\tilde \epsilon_{2,1}v_1 \\
\vdots \\
\tilde \epsilon_{N-1,N-2}v_{N-2} \\
\tilde \epsilon_{N,N-1}v_{N-1} 
\end{pmatrix} +\begin{pmatrix}
\tilde \epsilon_{1,1}v_1 \\
\tilde \epsilon_{2,2}v_2\\
\vdots \\
\tilde \epsilon_{N-1,N-1}v_{N-1}\\
\tilde \epsilon_{N,N}v_N
\end{pmatrix} +\begin{pmatrix}
\tilde \epsilon_{1,2}v_2\\
 \tilde\epsilon_{2, 3}v_3\\
\vdots \\
\tilde\epsilon_{N-1, N}v_N\\
0
\end{pmatrix},
\end{align*}
where $v=(v_1, \ldots, v_N)^\top$ with the superscript $\top$ denoting the transpose. 

Leveraging Theorem \ref{thm:weyltheorem1} then proves the statement.
\end{proof}

\begin{remark}
Due to the fact that the matrix $T_N^{(a,b)}$ is tridiagonal and the Hermitian matrix $A$ is stable, we do not require any constraint on the perturbation level in Theorem \ref{thm:eigenvaluestability0}. When it comes to banded Toeplitz matrices and complex perturbations, the findings and proofs are more intricate. We refer the readers to \cite{sjostrand2016large, sjostrand2021toeplitz} for further information on this subject.
\end{remark}


Applying Theorem \ref{thm:eigenvaluestability0} to the eigenvalues of $T_N^{(\eta, \beta)}$ in Lemma \ref{lem:eigenvalues0} yields the following stability estimate.
\begin{theorem}\label{thm:eigenvaluestability1}
For the eigenvalue $\widehat\lambda_{k}$ of the perturbed Toeplitz matrix $\widehat T_N^{(\eta, \beta)}$ with 
\begin{equation}\label{equ:eigenvaluestability1}
\max_{j=1,\dots,N}\left(\babs{\epsilon_{\eta, j}}, \babs{\epsilon_{\alpha, j}}, \babs{\epsilon_{\beta, j}}\right)\eqqcolon \epsilon, 
\end{equation}
 we have $\widehat \lambda_1 = \epsilon_1$ and
    \begin{equation}\label{equ:eigenvaluestability3}
    \widehat\lambda_{k} = \lambda_{k}+\epsilon_k = \alpha+2\sqrt{\eta\beta}\cos\left(\frac{(k-1)\pi}{N}\right)+ \epsilon_k, \quad k=2, \cdots, N,
    \end{equation}
    with $\babs{\epsilon_k}\leq C_1(\eta, \beta, \epsilon)\epsilon, 1\leq k\leq N$ and $C_1(\eta, \beta, \epsilon)$ being defined by (\ref{equ:eigenvaluestability2}). In particular, all the ${\widehat \lambda}_k$'s are real numbers. 
\end{theorem}

\begin{remark}
The eigenvalues of the tridiagonal Toeplitz matrix $T_N^{(0,0)}$ have been proved \cite{gregory1969collection, elliott1953characteristic} to be 
\[
\lambda_k = \alpha+2\sqrt{\eta \beta} \cos \left(\frac{k\pi}{N+1}\right), k=1, \cdots, N.
\]
\end{remark}

\begin{remark}
One can apply Theorems \ref{thm:eigenvaluestability0} and \ref{thm:eigenvectorstability1} to derive similar stability results for $T_N^{(0,0)}$. This corresponds to many examples in the non-Hermitian skin effect in condensed matter theory and quantum mechanics and thus our stability results here can be immediately applied to those examples. For the eigenvalues of the tridiagonal Toeplitz matrix with various perturbations on the corners, we refer the reader to \cite{yueh2005eigenvalues, yueh2008explicit}. 
\end{remark}

\section{Stability of eigenvectors}\label{sect3}
This section is devoted to estimating the stability of the eigenvectors of $\capmatg = T_N^{(\eta, \beta)}$.\\
For $\lambda_k$ defined in (\ref{equ:eigenvalueequ1}), let 
\begin{equation}\label{equ:defiofpj1}
p_{j}(\lambda_{k}) =  \left(\eta \sin \left(\frac{(j+1)(k-1) \pi}{N}\right)-\eta \sqrt{\frac{\eta}{\beta}} \sin \left(\frac{j(k-1) \pi}{N}\right)\right), \quad j=0,\cdots, N-1.
\end{equation}
Note that 
\begin{equation}\label{equ:boundonp1}
\babs{p_j(\lambda_{k})}\leq \babs{\eta}\left(1+\sqrt{\frac{\eta}{\beta}}\right), \quad j=0, \cdots, N-1. 
\end{equation}
The following results hold. 
\begin{theorem}\label{thm:eigenvectorstability1}
For $\widehat T_{N}^{(\eta, \beta)}$ defined by (\ref{equ:toeplitzallperturbed1}) and satisfying (\ref{equ:eigenvaluestability1}) and its eigenvalues $\widehat\lambda_k=\lambda_{k}+\epsilon_k, k=2,\cdots, N$ defined by (\ref{equ:eigenvaluestability3}), the corresponding eigenvectors are given by
\begin{equation}\label{equ:eigenvectorperturb1}
    \begin{aligned}
        \mathbf{\widehat x}_k =& \left(p_{0}(\lambda_{k})+\delta_0(\lambda_k),s \left(p_{1}(\lambda_{k})+\delta_{1}(\lambda_k)\right), s^2\left(p_2(\lambda_k)+\delta_{2}(\lambda_k)\right),  \cdots,\right. \\
        &\qquad \left. s^{N-1} \left(p_{N-1}(\lambda_k)+\delta_{N-1}(\lambda_k)\right) \right)^{\top},\quad k=2,\cdots, N,
    \end{aligned}
\end{equation}
where $s=\sqrt{\frac{\eta}{\beta}}$ and $p_{j}(\lambda_k)$ is defined in (\ref{equ:defiofpj1}). Moreover, we have 
	\begin{equation}\label{equ:eigenvectorstability1}
	\babs{ (-s)^j p_{j}(\lambda_k)}\leq \left(\frac{\eta}{\beta}\right)^{\frac{j}{2}}\babs{\eta}\left(1+\sqrt{\frac{\eta}{\beta}}\right), \quad j=0,1,\cdots, N-1,
	\end{equation}
 and
	\begin{equation}\label{equ:eigenvectorstability2}
	\babs{s^{j}\delta_j(\lambda_k)}\leq \zeta_{k,j}  \epsilon,\quad j=0,1,\cdots, N-1,
	\end{equation}
where
	\[
\zeta_{k,j} = \left(\sqrt{\frac{\eta}{\beta}}\left(\frac{\babs{\beta}(\babs{\eta}+\epsilon)}{(\babs{\beta}-\epsilon)|\eta|}\right)\right)^j\left(a_+r_{k,+}^{j}+a_-r_{k,-}^j-\zeta\right) 
	\]
 with 
\begin{equation}\label{equ:rpmformula1}
r_{k,\pm} = \left(\babs{\cos\left(\frac{(k-1)\pi}{N}\right)}+\frac{C_2(\eta, \beta)\epsilon}{\sqrt{\eta \beta }}\right)\pm\sqrt{\left(\babs{\cos\left(\frac{(k-1)\pi}{N}\right)}+\frac{C_2(\eta, \beta)\epsilon}{\sqrt{\eta \beta }}\right)^2+1},
\end{equation}
$C_2(\eta, \beta, \epsilon) = \frac{\babs{\beta}+\babs{\eta}+\epsilon}{2\sqrt{\beta \eta}}+1$, and $a_+, a_-, \zeta$ being bounded constants. In particular, for those indices $k$ such that 
\begin{equation}\label{equ:eigenvectorstability3}
\babs{\sqrt{\frac{\eta}{\beta}}\left(\frac{\babs{\beta}(\babs{\eta}+\epsilon)}{(\babs{\beta}-\epsilon)|\eta|}\right)r_{k,+}}<1,
\end{equation}
the corresponding eigenvector still has an exponential decay. 


As a result, there exists a constant $c$ such that if $\sqrt{\eta/\beta}<\sqrt{2}-1$ and $\epsilon< \frac{c}{N^2}$, then we still have exponential decay for all the corresponding eigenvectors $\mathbf {\widehat x}_k, 2\leq k \leq N,$ of $\widehat T_{N}^{(a, b)}$. Further, if we require $\sqrt{\eta/\beta}$ to be even smaller, then this exponential decay will remain for even larger values of $\varepsilon$.
\end{theorem}
\begin{proof}
The inequality (\ref{equ:eigenvectorstability1}) is a direct consequence of (\ref{equ:eigenvectorequ1}) and (\ref{equ:boundonp1}). The rest of the argument consists in proving (\ref{equ:eigenvectorstability2}). For fixed $\lambda_k$, we see that
\begin{equation}\label{equ:proofeigenvectorstability0}
	\babs{\delta_j(\lambda_k)}\leq \left(\frac{\babs{\beta}(\babs{\eta}+\epsilon)}{(\babs{\beta}-\epsilon)|\eta|}\right)^{j} M_j \epsilon
	\end{equation}
and will analyse the constants $M_j$. We consider the eigenvalue problem
	\begin{equation}\label{equ:proofeigenvectorstability1}
	\left(\widehat{T}_{N}^{(\eta, \beta)}-\widehat \lambda_{k} I\right) \mathbf {\widehat x}_k=0,
	\end{equation}
	where $\mathbf {\widehat x}_k$ is (\ref{equ:eigenvectorperturb1}). For simplicity, we abbreviate $p_j(\lambda_{k})$ as $p_j$ and $\delta_j(\lambda_k)$ as $\delta_j$ in the proof. Based on the first row in (\ref{equ:proofeigenvectorstability1}), we have 
	\[
	(\alpha +\eta+\epsilon_{\alpha,1}-\lambda_k -\epsilon_k)(p_0+\delta_0)+(\beta+\epsilon_{\beta, 1})s(p_1+\delta_{1}) = 0.
	\]
	This gives 
	\[
	(\epsilon_{\alpha,1}-\epsilon_k)p_0 +\epsilon_{\beta, 1}sp_1 + (\alpha +\eta+\epsilon_{\alpha,1}-\lambda_k - \epsilon_k)\delta_0+(\beta+\epsilon_{\beta, 1})s \delta_{1} =0,
	\]
	where we have used $(\alpha+\eta-\lambda_k)p_0+\beta sp_1=0$. Let $\delta_0=\epsilon$. Then 
	\[
	M_{0} =1
	\]
	and 
	\[
	\delta_1 = \frac{(\epsilon_{\alpha,1}-\epsilon_k)p_0 +\epsilon_{\beta, 1}sp_1+(\alpha +\eta+\epsilon_{\alpha,1}-\lambda_k-\epsilon_k)\epsilon }{-(\beta+\epsilon_{\beta, 1})s}.
	\]
 Denote $C_2(\eta, \beta, \epsilon)= \frac{\babs{\beta}+\babs{\eta}+\epsilon}{2\sqrt{\beta \eta}}+1$. Then we have
	\begin{align*}
	\babs{\delta_1} \leq & \babs{\frac{(\epsilon_{\alpha,1}-\epsilon_k)p_0 +\epsilon_{\beta, 1}sp_1}{-(\beta+\epsilon_{\beta, 1})s}}+ \babs{\frac{(\alpha +\eta+\epsilon_{\alpha,1}-\lambda_k-\epsilon_k)\epsilon }{-(\beta+\epsilon_{\beta, 1})s}}\\
	\leq & \frac{|\eta|(1+s)(2C_2(\eta, \beta)+s) \epsilon}{(|\beta|-\epsilon)s}+\frac{\babs{\eta}+2\sqrt{\eta \beta}\babs{\cos\left(\frac{(k-1)\pi}{N}\right)}+2C_2(\eta, \beta, \epsilon)\epsilon}{(\babs{\beta} -\epsilon)s}\epsilon\\
	\leq & \left(\frac{\babs{\beta}(\babs{\eta}+\epsilon)}{(\babs{\beta}-\epsilon)|\eta|}\right)\left(s(1+s)(2C_2(\eta, \beta, \epsilon)+s)+ s+2\babs{\cos\left(\frac{(k-1)\pi}{N}\right)}+\frac{2C_2(\eta, \beta, \epsilon)\epsilon}{\sqrt{\eta \beta }}\right) \epsilon,
	\end{align*}
where we have used (\ref{equ:boundonp1}) and Theorem \ref{thm:eigenvaluestability1} in the second inequality. Therefore,
\[
M_1 = s(1+s)(2C_2(\eta, \beta, \epsilon)+s)+ s+ 2\babs{\cos\left(\frac{(k-1)\pi}{N}\right)}+\frac{2C_2(\eta, \beta, \epsilon)\epsilon}{\sqrt{\eta \beta }}.
\]

We now analyse the relation between $M_{j-1}, M_{j}$ and $M_{j+1}$. Based on the $(j+1)$-th row of (\ref{equ:proofeigenvectorstability1}), we have 
	\[
(\eta+\epsilon_{\eta, j+1})s^{j-1}(p_{j-1}+\delta_{j-1}) + (\alpha+\epsilon_{\alpha, j+1}-\lambda_k-\epsilon_k)s^{j}(p_{j}+\delta_j)+(\beta+\epsilon_{\beta, j+1})s^{j+1}(p_{j+1}+\delta_{j+1}) =0,
\]
which yields
\[
(\eta+\epsilon_{\eta, j+1})(p_{j-1}+\delta_{j-1}) +(\alpha+\epsilon_{\alpha, j+1}-\lambda_k-\epsilon_k)s(p_{j}+\delta_j)+(\beta+\epsilon_{\beta, j+1})s^{2}(p_{j+1}+\delta_{j+1}) =0.
\]
Making use of the identity $$\eta p_{j-1}+(\alpha -\lambda_k)s p_j+\beta s^2p_{j+1}=0,$$ we can eliminate some items to arrive at 
\begin{align*}
	&\epsilon_{\eta, j+1}p_{j-1} + (\epsilon_{\alpha, j+1}-\epsilon_k)sp_{j}+ \epsilon_{\beta, j+1} s^2 p_{j+1} \\
	&\qquad   +(\eta+\epsilon_{\eta, j+1})\delta_{j-1} +(\alpha+\epsilon_{\alpha, j+1}-\lambda_k-\epsilon_k)s\delta_{j}+ (\beta+\epsilon_{\beta, j+1})s^2\delta_{j+1} =0.
\end{align*}
Therefore, 
\begin{align*}
&\delta_{j+1} \\
=& \frac{\epsilon_{\eta, j+1}p_{j-1} +(\epsilon_{\alpha, j+1}-\epsilon_k)sp_{j}+ \epsilon_{\beta, j+1} s^2 p_{j+1}+(\eta+\epsilon_{\eta, j+1})\delta_{j-1} +(\alpha+\epsilon_{\alpha, j+1}-\lambda_k-\epsilon_k)s\delta_{j}}{-(\beta+\epsilon_{\beta, j+1})s^2}\\
=& \frac{\epsilon_{\eta, j+1}p_{j-1} + (\epsilon_{\alpha, j+1}-\epsilon_k)sp_{j}+ \epsilon_{\beta, j+1} s^2 p_{j+1}}{-(\beta+\epsilon_{\beta, j+1})s^2}\\
&\quad + \frac{(\eta+\epsilon_{\eta, j+1})\delta_{j-1}}{-(\beta+\epsilon_{\beta, j+1})s^2} + \frac{(\alpha+\epsilon_{\alpha, j+1}-\lambda_k-\epsilon_k)s\delta_{j}}{-(\beta+\epsilon_{\beta, j+1})s^2} \\
:=& I_1 +I_2 +I_3.
\end{align*}
For the term $I_1$, by (\ref{equ:boundonp1}) we have 
\begin{align*}
&\babs{\frac{\epsilon_{\eta, j+1}p_{j-1} - (\epsilon_{\alpha, j+1}-\epsilon_k)sp_{j}+ \epsilon_{\beta, j+1} s^2 p_{j+1}}{-(\beta+\epsilon_{\beta, j+1})s^2}}\leq \frac{\epsilon(1+2C_2(\eta, \beta, \epsilon)s+s^2)\babs{\eta}\left(1+s\right)}{(\babs{\beta} -\epsilon)s^2}\\
\leq& \frac{\babs{\beta}}{\babs{\beta}-\epsilon}\left(1+2C_2(\eta, \beta, \epsilon)s+s^2\right)(1+s) \epsilon =: \frac{\babs{\beta}}{\babs{\beta}-\epsilon} C_3(\eta, \beta, \epsilon) \epsilon.
\end{align*}
For the term $I_2$, by (\ref{equ:proofeigenvectorstability0}) we obtain that
\[
\babs{\frac{(\eta+\epsilon_{\eta, j+1})\delta_{j-1}}{-(\beta+\epsilon_{\beta, j+1})s^2}}\leq \frac{\babs{\eta}+\epsilon}{(\babs{\beta} -\epsilon)s^2}\left(\frac{\babs{\beta}(\babs{\eta}+\epsilon)}{(\babs{\beta}-\epsilon) \babs{\eta}}\right)^{j-1}M_{j-1} \epsilon\leq \left( \frac{\babs{\beta}(\babs{\eta}+\epsilon)}{(\babs{\beta}-\epsilon) \babs{\eta}}\right)^{j+1}M_{j-1}\epsilon.
\]
For the term $I_3$, by Theorem \ref{thm:eigenvaluestability1} and estimate (\ref{equ:proofeigenvectorstability0}), it follows that
\begin{align*}
&\babs{\frac{(\alpha+\epsilon_{\alpha, j+1}-\lambda_k-\epsilon_k)s\delta_{j}}{-(\beta+\epsilon_{\beta, j+1})s^2} }\leq \frac{2\sqrt{\eta \beta}\babs{\cos\left(\frac{(k-1)\pi}{N}\right)}+2C_2(\eta, \beta, \epsilon)\epsilon}{(\babs{\beta} -\epsilon)s}\babs{\delta_{j}}\\
\leq& \frac{\babs{\beta}}{\babs{\beta}-\epsilon} \left(2\babs{\cos\left(\frac{(k-1)\pi}{N}\right)}+\frac{2C_2(\eta, \beta, \epsilon)\epsilon}{\sqrt{\eta \beta }}\right) \left(\frac{\babs{\beta}(\babs{\eta}+\epsilon)}{(\babs{\beta}-\epsilon)|\eta|}\right)^{j}M_j \epsilon.
\end{align*}
Therefore, we have 
\begin{align*}
\babs{\delta_{j+1}}\leq  \left(\frac{\babs{\beta}(\babs{\eta}+\epsilon)}{(\babs{\beta}-\epsilon)|\eta|}\right)^{j+1} \left(2\left(\babs{\cos\left(\frac{(k-1)\pi}{N}\right)}+\frac{C_2(\eta, \beta, \epsilon)\epsilon}{\sqrt{\eta \beta }}\right)M_j+ M_{j-1}+ C_3(\eta, \beta, \epsilon) \right)\epsilon.
\end{align*}
By (\ref{equ:proofeigenvectorstability0}), the recurrence relation is
\begin{equation}\label{equ:proofvectorstablity1}
M_{j+1} = 2\left(\babs{\cos\left(\frac{(k-1)\pi}{N}\right)}+\frac{C_2(\eta, \beta, \epsilon)\epsilon}{\sqrt{\eta \beta }}\right)M_j+ {M}_{j-1}+ C_3(\eta, \beta, \epsilon).
\end{equation}
Define 
\begin{equation}\label{equ:proofvectorstablity2}
\tilde{M_{j}} = M_j+\zeta
\end{equation}
with 
\[
\zeta = \frac{C_3(\eta, \beta, \epsilon)}{2\left(\babs{\cos\left(\frac{(k-1)\pi}{N}\right)}+\frac{C_2(\eta, \beta, \epsilon)\epsilon}{\sqrt{\eta \beta }}\right)}. 
\]
Then (\ref{equ:proofvectorstablity1}) yields
\[
\tilde{M}_{j+1}= 2\left(\babs{\cos\left(\frac{(k-1)\pi}{N}\right)}+\frac{C_2(\eta, \beta, \epsilon)\epsilon}{\sqrt{\eta \beta }}\right)\tilde{M}_{j}+\tilde{M}_{j-1}.
\]
This is a linear recurrence relation. To solve it, we consider the roots of its characteristic equation
\[
r_k^2 = 2\left(\babs{\cos\left(\frac{(k-1)\pi}{N}\right)}+\frac{C_2(\eta, \beta, \epsilon)\epsilon}{\sqrt{\eta \beta }}\right)r_k+1,
\]
which are given by (\ref{equ:rpmformula1}). Thus we have 
\[
\tilde{M}_j = a_+r_{k,+}^{j}+a_-r_{k,-}^j
\]
with $a_{+}, a_{-}$ being  chosen to satisfy the initial conditions $$\tilde{M}_0=M_0+\zeta, \tilde{M}_1=M_1+\zeta.$$ It is not hard to see that $a_{+}, a_{-}$ are bounded. Now, by (\ref{equ:proofvectorstablity2}), $M_j$ reads 
\[
M_j = a_+r_{k,+}^{j}+a_-r_{k,-}^j-\zeta, 
\]
and we thus have 
\[
\babs{\delta_j}\leq\left(\frac{\babs{\beta}(\babs{\eta}+\epsilon)}{(\babs{\beta}-\epsilon)|\eta|}\right)^{j} \left(a_+r_{k,+}^{j}+a_-r_{k,-}^j-\zeta\right)\epsilon.
\]
This proves (\ref{equ:eigenvectorstability2}).
\end{proof}



Some remarks are now in order. 
\begin{remark}
The stability results obtained in Sections \ref{sect2} and \ref{sect3} can be generalised to the dimer case where one can utilise the characterization of the eigenvalues and eigenvectors given in \cite{dimerSkin}. 
\end{remark}

\begin{remark}
Note that the perturbation to the spacings $\{s_i\}$ between subwavelength resonators or the coefficient $\gamma$ in (\ref{equ:nonhermitiancoeff1}) of order $\epsilon$ will result in an $O(\epsilon)$ perturbation in the nonzero entries of the gauge capacitance matrix $\capmatg$. Thus, Theorem \ref{thm:eigenvectorstability1} can be applied to $\hat \capmat^{\gamma}$ to obtain a stability estimate to the eigenvectors and the skin effect of $\capmatg$. 
\end{remark}

We can illustrate numerically the results stated in Theorem \ref{thm:eigenvectorstability1}. In particular, we consider typical values in physical applications.  We let $\eta=0.15$, $\beta=3.15$ (which correspond to $\ell=s=1$ and $\gamma=3$) and $\epsilon$ satisfying (\ref{equ:eigenvectorstability3}). The results are presented in \cref{fig: validation theoreical result}, where we show the eigenvectors of a system of 50 subwavelength resonators on a logarithmic axis. If the perturbations are sufficiently small that the condition \eqref{equ:eigenvectorstability3} is satisfied, then the eigenvectors still all have the $(\sqrt{\beta/\eta})$ decay rate. However, when the perturbations are large enough that condition \eqref{equ:eigenvectorstability3} does not hold for some indices, then the corresponding modes have a much lower decay rate.

\begin{figure}[h]
    \centering
    \begin{subfigure}[t]{0.48\textwidth}
    \centering
    \includegraphics[height=0.76\textwidth]{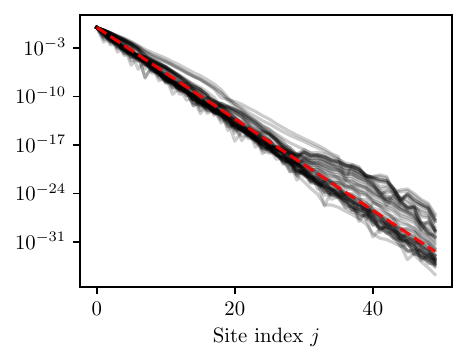}
    \caption{Exponential decay of the eigenvectors for $\epsilon$ satisfying \eqref{equ:eigenvectorstability3}. The eigenvectors superimposed on one another on a semi-log plot. The red dashed line represents $(\sqrt{\beta/\eta})^j$. We observe the same decay rate as the unperturbed case.}
    \label{fig: exponential decay small eps}
    \end{subfigure}\hfill
    \begin{subfigure}[t]{0.48\textwidth}
    \includegraphics[height=0.75\textwidth]{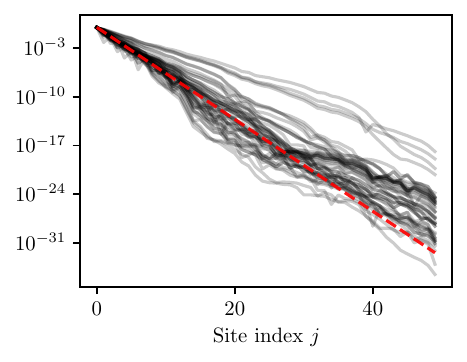}
        \caption{Decay of the eigenvectors for $\epsilon$ \emph{not} satisfying \eqref{equ:eigenvectorstability3}. The eigenvectors superimposed on one another on a semi-log plot. The red dashed line represents $(\sqrt{\beta/\eta})^j$. We observe several eigenvectors with lower decay rate than the one in Figure \ref{fig: exponential decay small eps}.}
    \label{fig: exponential decay big eps}
    \end{subfigure}
    \caption{Numerical illustration of the stability of the eigenvector decay rate predicted by \cref{thm:eigenvectorstability1}. The Toeplitz matrix has coefficients $\eta=0.15$, $\beta=3.15$ and is of size $50\times 50$.}
    \label{fig: validation theoreical result}
\end{figure}

\section{Numerical illustrations} \label{sect4}
In this section, we provide numerical evidence of the stability of the non-Hermitian skin effect and show how it competes with Anderson-type localisation of the eigenmodes in the bulk when the disorder is large. We will consider perturbations in both the geometry and the local values of the imaginary gauge potential. For the sake of brevity, we fix the size $\ell$ of the resonators  and perturb independently either $\gamma$ or the spacing $s$ between the resonators.

\subsection{Random perturbations of the geometry}
We first consider systems of subwavelength resonators where the relative spacings are perturbed as
\begin{align}
    {s_i} = 1 + \epsilon_i, \qquad \epsilon_i \sim \mathcal{U}_{[-\varepsilon,\varepsilon]}.
\end{align}
Here, $\mathcal{U}_{[-\varepsilon,\varepsilon]}$ is a uniform distribution with support in $[-\varepsilon,\varepsilon]$. In \cref{fig: condensation and winding}, we study how the eigenmodes of a system of 30 subwavelength resonators behave as the disorder increases. These results are averages based on 500 independent realisations. We show the relative proportion of eigenvalues that fall within the region of negative winding of the associated Toeplitz operator from \cref{fig: winding region}, as well as the proportion of eigenmodes accumulating at the left edge (which for this and the following figures has been defined as the number of eigenvectors that attain their maximal value, in absolute terms, in one of the first two dimers). We consider values of the disorder strength that are small enough that the resonators are guaranteed to not overlap. Both these quantities are constant for small disorder strengths then decrease once the disorder strength passes a certain threshold (as predicted by Theorem~\ref{thm:eigenvectorstability1}). The intersection of these two sets is also shown.

One notices very similar trends in the three lines in \cref{fig: condensation and winding}, with small differences due to the imperfect formulation of the accumulation measure and the perturbations. On the other hand, \cref{fig: localisation} shows the localisation of the eigenvectors for different disorder strengths. The localisation of the eigenvectors is measured using the quantity $\Vert v_i\Vert_\infty / \Vert v_i\Vert_2$ and the different lines correspond to different disorder strengths $\varepsilon$. We notice that the lines are indistinguishable, indicating that the localisation of the eigenvectors is independent of any random perturbation of the positions of the resonators.

\cref{fig: phase change and topological protection} shows similar stability properties  as those in \cref{fig: condensation and winding}, but here the relative number of eigenvalues falling within the region with negative winding is plotted for different values of $\varepsilon$ and $\gamma$. On the left side of the figure we see the topologically protected region: for these values of $\gamma$ any small perturbation size $\varepsilon$ will not cause any eigenvalue to exit the region and thus the corresponding eigenvector remains accumulated at the left edge of the structure.

The results in \cref{fig: geometry perturbation} show how the proportion of eigenvectors localised to the left edge of the system decreases as the disorder increases. Studying in the eigenvectors themselves, as shown in \cref{fig: condensed_single_realisations} for three different values of the disorder strength, we see that increasing disorder means an increasing number of eigenvectors are localised in the bulk rather than on the left edge. This behaviour is typical of Anderson-type localisation in disordered systems and demonstrates the internal competition between the skin effect and Anderson localisation.

\begin{figure}[h]
    \centering
    \begin{subfigure}[t]{0.45\textwidth}
    \centering
\includegraphics[height=0.75\textwidth]{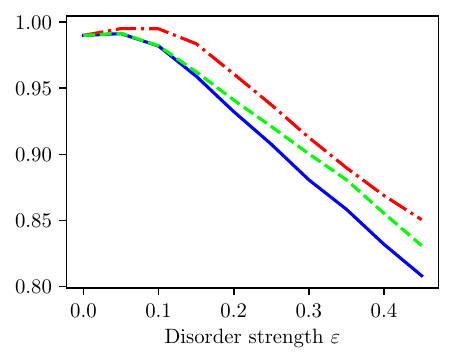}
    \caption{Eigenmode accumulation at one edge and topological winding. The green dashed line shows the average proportion of eigenvectors which are localised at the left edge. The red dash-dot line shows the the average proportion of eigenvalues that lay in the topologically protected region. The blue solid line shows the proportion of eigenpairs that have \emph{both} eigenvalues in the topologically protected region and eigenvectors accumulated on the left edge.}
    \label{fig: condensation and winding}
    \end{subfigure}\hfill
    \begin{subfigure}[t]{0.45\textwidth}
    \centering
    \includegraphics[height=0.76\textwidth]{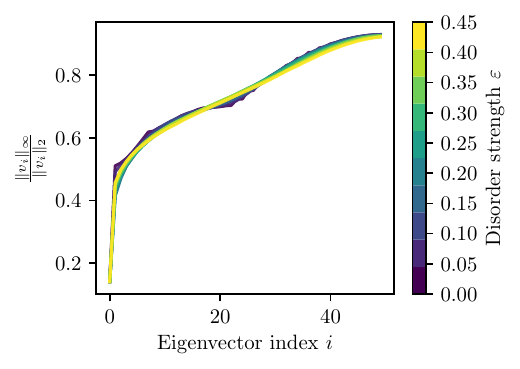}
    \caption{Eigenmode localisation. Each line shows the average eigenmode localisation for a different value of the disorder strength $\varepsilon$. For small $\epsilon$ the localisation is due to the skin effect, while for big $\epsilon$ it is consequence of the Anderderson localisation. As the lines are indistinguishable we conclude that the eigenmode localisation is independent of disorder strength; as $\varepsilon$ increases, modes might be localised in the bulk but will not become delocalised.}
    \label{fig: localisation}
    \end{subfigure}
    \\[5mm]
    \begin{subfigure}[t]{0.45\textwidth}
    \includegraphics[height=0.75\textwidth]{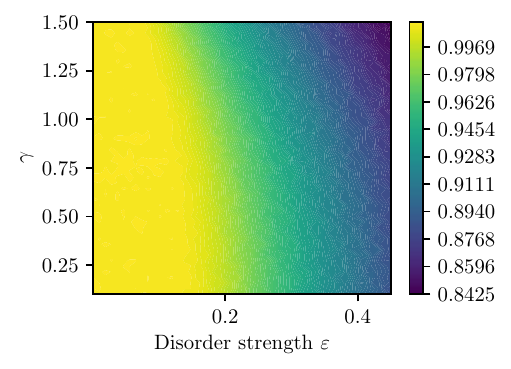}
        \caption{Phase change and topological protection. The color scale shows the average proportion of eigenvalues that lay in the topologically protected region for different values of $\gamma$. The left yellow zone is the stability region.}
    \label{fig: phase change and topological protection}
    \end{subfigure}
    \caption{Competition between the non-Hermitian skin effect and Anderson localisation when perturbing the geometry. The non-Hermitian skin effect shows stability with respect to random perturbations. Outside of the stability region, there is competition with Anderson localisation. Averages are computed over $500$ runs for a system of $50$ resonators with $\ell=s=1$.}
    \label{fig: geometry perturbation}
\end{figure}

\begin{figure}
    \centering
    \begin{subfigure}[t]{0.32\textwidth}
    \centering
    \includegraphics[height=0.76\textwidth]{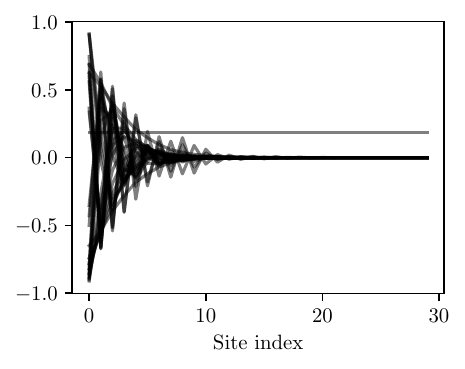}
    \caption{Single realisation with disorder strength $\varepsilon=0.1$. All eigenmodes are accumulated on the left edge.}
    \label{fig: condensed_single_realisation1}
    \end{subfigure}\hfill
    \begin{subfigure}[t]{0.32\textwidth}
    \centering
    \includegraphics[height=0.76\textwidth]{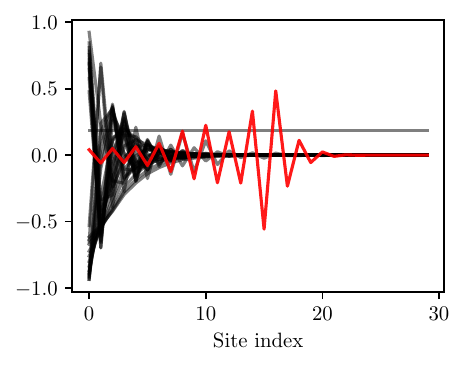}
    \caption{Single realisation with disorder strength $\varepsilon=0.2$. One eigenmode localised in the bulk is highlighted in red.}
    \label{fig: condensed_single_realisation2}
    \end{subfigure}\hfill
    \begin{subfigure}[t]{0.32\textwidth}
    \centering
    \includegraphics[height=0.76\textwidth]{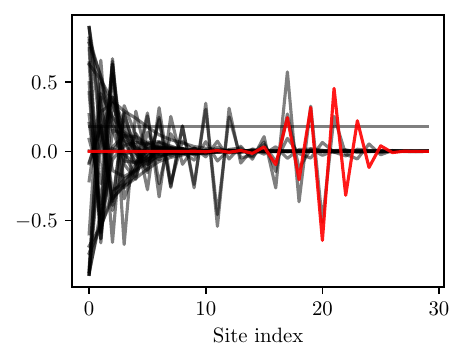}
    \caption{Single realisation with disorder strength $\varepsilon=0.4$. One eigenmode localised in the bulk is highlighted in red.}
    \label{fig: condensed_single_realisation4}
    \end{subfigure}
    \caption{Stability of the non-Hermitian skin effect under perturbations of the geometry. Eigenmode condensation on the left edge of the structure with some eigenmodes localised in the bulk. Single realisations with  $N=30$, $s=\ell=1$, and $\varepsilon=0.1,0.2,0.4$ for \cref{fig: condensed_single_realisation1}, \ref{fig: condensed_single_realisation2}, and \ref{fig: condensed_single_realisation4} respectively. This should be compared with \cref{fig: condendensed eigenvectors unperturbed}, where there is no disorder.}
    \label{fig: condensed_single_realisations}
\end{figure}

 \subsection{Random perturbations of the imaginary gauge potential}

In this subsection we consider systems of subwavelength resonators where now the spacing between the resonators is fixed to $s_i=1$, but the damping factor $\gamma$ is allowed to be different in each resonator. Specifically, we consider
\begin{align}
    {\gamma_i} = 1 + \epsilon_i, \qquad \epsilon_i \sim \mathcal{U}_{[-\varepsilon,\varepsilon]},
\end{align}
where $\gamma_i$ is the value taken by $\gamma$ in the $i$-th resonator. 

\cref{fig: gamma perturbation} is the analogue of \cref{fig: geometry perturbation} in this case. In this case, the disorder strength $\varepsilon$ is allowed to vary over a larger range as we do not have the issue of resonators overlapping. Note however that large disorder, such as $\vert \epsilon \vert > \vert \gamma \vert=1$, will possibly induce different signs in the $\gamma_i$ and thus striking changes in the coefficients of $\mathcal{C}^\gamma$. \cref{fig: condensation and winding pert_gamma} shows some similar behaviour to \cref{fig: condensation and winding}, in the sense that both quantities decrease as the disorder increases. However, for larger values of $\varepsilon$ there is an obvious decoupling of the quantities. This is due to the fact that, for very large random perturbations, $\widehat{\mathcal{C}}^\gamma$ is very far from being Toeplitz and thus the symbol of the associated Toepliz operator loses its meaning.

On the other hand, \cref{fig: localisation pert_gamma} shows the localisation measure of the eigenvectors for different disorder strengths. Eventhough the disorder is allowed to take much larger values here than in \cref{fig: localisation}, we once again observe that the localisation of a given eigenvector is almost constant as the disorder changes.

Finally, \cref{fig: phase change and topological protection pert_gamma} shows the analogous results to \cref{fig: phase change and topological protection}. Once again, we see that there is a region of topological protection. This time, it is in the top-left of the diagram (for large space $s$ and small disorder $\varepsilon$).

\Cref{fig: gamma perturbation,fig: condensed_single_realisations,fig: geometry perturbation} as a whole show an internal competition between the skin effect and the Anderson localisation: as disorder is introduced, modes transition from being condensed on the edge to being localised within the bulk. 

\begin{figure}[h]
    \centering
    \begin{subfigure}[t]{0.45\textwidth}
    \centering
\includegraphics[height=0.75\textwidth]{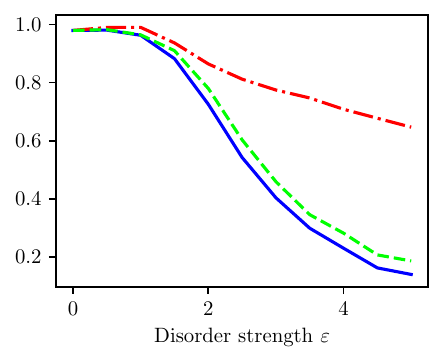}
    \caption{Eigenmode accumulation at one edge and topological winding. The green dashed line shows the average proportion of eigenvectors which are localised at the left edge. The red dash-dot line shows the the average proportion of eigenvalues that lay in the topologically protected region. The blue solid line shows the proportion of eigenpairs that have \emph{both} eigenvalues in the topologically protected region and eigenvectors accumulated on the left edge.}
    \label{fig: condensation and winding pert_gamma}
    \end{subfigure}\hfill
    \begin{subfigure}[t]{0.45\textwidth}
    \centering
    \includegraphics[height=0.76\textwidth]{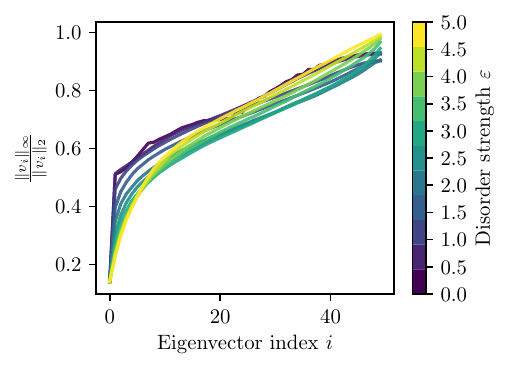}
    \caption{Eigenmode localisation. Each line shows the average eigenmode localisation for a different value of the disorder strength $\varepsilon$. For small $\epsilon$ the localisation is due to the skin effect, while for big $\epsilon$ it is consequence of the Anderderson localisation. Thus, the localisation is much less sensitive to the perturbations than the accumulation (as the position of localisation may be away from the edge for large disorder).}
    \label{fig: localisation pert_gamma}
    \end{subfigure}
    \\[5mm]
    \begin{subfigure}[t]{0.45\textwidth}
    \includegraphics[height=0.75\textwidth]{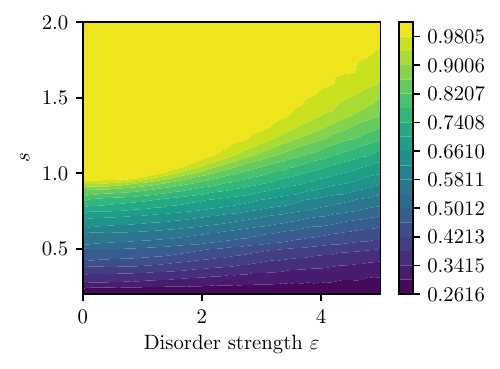}
    \caption{Phase change and topological protection. The color scale shows the average proportion of eigenvalues that lay in the topologically protected region for different values of $s$. The top-left yellow zone is the stability region.}
    \label{fig: phase change and topological protection pert_gamma}
    \end{subfigure}
    \caption{Competition between the non-Hermitian skin effect and Anderson localisation when perturbing the complex gauge potential. The non-Hermitian skin effect is stable with respect to random perturbations. Outside of the stability region,  there is a competition between the Non-Hermitian skin effect and the disorder-induced  Anderson localisation. Averages $500$ runs for a system of 50 resonators with $\ell=s=1$.}
    \label{fig: gamma perturbation}
\end{figure}

\subsection{Simultaneous perturbations of the geometry and the imaginary gauge potential}
For the sake of completeness, in \cref{fig: pert of gamma and s} we present the result of perturbing $\gamma$ and $s$ simultaneously by
\begin{align*}
    {s_i} &= 1 + \epsilon_i, \qquad \epsilon_i \sim \mathcal{U}_{[-\varepsilon_s,\varepsilon_s]}\\
    {\gamma_i} &= 1 + \epsilon_i,  \qquad\epsilon_i \sim \mathcal{U}_{[-\varepsilon_\gamma,\varepsilon_\gamma]}.
\end{align*}
The results show that the skin effect is very stable under any type of perturbations: the spacing between the resonators may be perturbed up to roughly $10\%$ and simultaneously the $\gamma$ factor up to $50\%$ independently in every resonator and the accumulation of eigenmodes on one edge of the structure remains unaltered.

\begin{figure}
    \centering
    \includegraphics[width=0.5\textwidth]{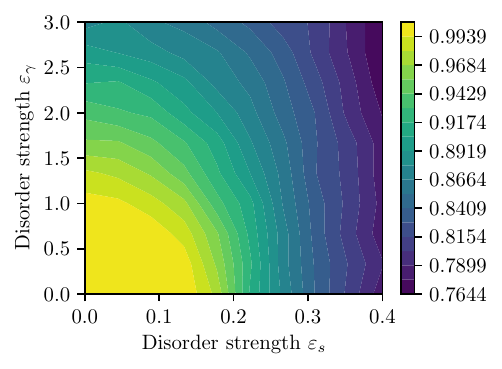}
    \caption{Phase transition for varying disorder strengths in $\gamma$ and $s$. The color scale shows the average proportion of eigenvalues that lay in the topologically protected region for different values of $\varepsilon_\gamma$ and $\varepsilon_s$. The bottom-left yellow zone is the stability region. The values are averages over 100 runs in a system of $50$ resonators with $\ell=1$.}
    \label{fig: pert of gamma and s}
\end{figure}

\section{Concluding remarks} \label{sect5}
Based on a stability analysis of the eigenvalues and eigenvectors of the gauge capacitance matrix, we have proved robustness of the non-Hermitian skin effect with respect to random changes of the strength $\gamma$ of the imaginary gauge potential and the spacing $s$ between the resonators. We have also elucidated the topological origins of such robustness in our setting. Under random perturbations, the eigenmodes which remain localised at the edge of the structure are precisely those whose associated eigenvalues (which remain real valued) remain within the region of the complex plane corresponding to negative winding of the symbol of the corresponding Toeplitz operator. As the strength of the disorder increases, an increasing number of eigenmodes become localised in the bulk as their corresponding eigenfrequencies leave the region of negative winding. This leads to a competition between the non-Hermitian skin effect and Anderson localisation in the bulk.  

The results in this paper could be generalised to systems with periodically repeated cells of $K \geq 2$ resonators \cite{ammari.li.ea2023Mathematical} and to higher dimensional systems, in which it is well known that the skin effect can be realised \cite{kai,3DSkin}. Since our results are based on an asymptotic matrix model for subwavelength physics, they can also be generalised to analogous tight-binding models in condensed matter theory. 

\addtocontents{toc}{\protect\setcounter{tocdepth}{0}}
\section*{Acknowledgments}
The work of PL was supported by Swiss National Science Foundation grant number 200021–200307. The work of BD was supported by a fellowship funded by the Engineering and Physical Sciences
Research Council (EPSRC) under grant number EP/X027422/1. 

\bigskip

\section*{Code availability}
The software used to produce the numerical results in this work is openly available at \\ \href{https://doi.org/10.5281/zenodo.8210678}{https://doi.org/10.5281/zenodo.8210678}.





\printbibliography

\end{document}

\begin{align*}
s^{j+1}\delta_{j+1} =& \frac{\epsilon_{\eta, j+1}s^{j-1}p_{j-1} + (\epsilon_{\alpha, j+1}-\epsilon_k)s^jp_{j}+ \epsilon_{\beta, j+1} s^{j+1} p_{j+1}}{-(\beta+\epsilon_{\beta, j+1})}\\
&\quad + \frac{(\eta+\epsilon_{\eta, j+1})s^{j-1}\delta_{j-1}}{-(\beta+\epsilon_{\beta, j+1})} + \frac{(\alpha+\epsilon_{\alpha, j+1}-\lambda_k-\epsilon_k)s^j\delta_{j}}{-(\beta+\epsilon_{\beta, j+1})} \\
&:= I_1 +I_2 +I_3.
\end{align*}
\begin{align*}
\babs{I_1}\leq s^{j-1}\frac{C_2(\gamma, \beta)\epsilon}{\babs{\beta}-\epsilon}
\end{align*}

%% file: comands.tex
\usepackage[T1]{fontenc} 
\usepackage{lmodern} 
\rmfamily 
\DeclareFontShape{T1}{lmr}{b}{sc}{<->ssub*cmr/bx/sc}{}
\DeclareFontShape{T1}{lmr}{bx}{sc}{<->ssub*cmr/bx/sc}{}
\usepackage{lipsum}
\usepackage{mathtools}
\usepackage{amsmath}
\usepackage{amssymb}
\usepackage{tikz}
\usetikzlibrary{matrix,positioning,decorations.pathreplacing,calc}
\usetikzlibrary{
decorations.text,%
decorations.markings,%
shadows}
\usepackage{adjustbox}
\usepackage{graphicx}

\usepackage{dsfont} 
\usepackage[format=hang]{caption}
\usepackage{subcaption}
\usepackage[normalem]{ulem}

\newcommand{\abs}[1]{\lvert#1\rvert}
\usepackage{bm}

\usepackage[
style=numeric,
sorting=nty,
maxnames=99,
maxalphanames=5,
natbib=true,
sortcites]{biblatex}

\DeclareNameAlias{default}{family-given}

\graphicspath{{figures/}}

\AtEveryBibitem{
 \clearfield{url}
 \clearfield{issn}
 \clearfield{isbn}
 \clearfield{urldate}
 
 \ifentrytype{book}{
  \clearfield{pages}}{%
 }
}

\usepackage{xargs}                      
\usepackage[prependcaption,textsize=tiny,textwidth=3cm]{todonotes}
\newcommandx{\unsure}[2][1=]{\todo[linecolor=red,backgroundcolor=red!25,bordercolor=red,#1]{#2}}
\newcommandx{\change}[2][1=]{\todo[linecolor=blue,backgroundcolor=blue!25,bordercolor=blue,#1]{#2}}
\newcommandx{\info}[2][1=]{\todo[linecolor=OliveGreen,backgroundcolor=OliveGreen!25,bordercolor=OliveGreen,#1]{#2}}
\newcommandx{\improvement}[2][1=]{\todo[linecolor=black,backgroundcolor=black!25,bordercolor=black,#1]{#2}}
\newcommandx{\thiswillnotshow}[2][1=]{\todo[disable,#1]{#2}}
\usepackage{amsthm}
\usepackage[noabbrev,capitalize]{cleveref}
\crefname{proposition}{Proposition}{Propositions}
\crefname{equation}{}{}

\newtheorem{theorem}{Theorem}[section]
\newtheorem{lemma}[theorem]{Lemma}

\theoremstyle{definition}

\newtheorem{remark}[theorem]{Remark}

\crefname{assumption}{Assumption}{Assumptions}
\crefname{definition}{Definition}{Definitions}
\crefname{corollary}{Corollary}{Corollaries}
\crefname{enumi}{item}{items}

\usepackage{fancyhdr}

\fancypagestyle{plain}{%
\fancyhead[C]{}
\fancyfoot[C]{}

}
\fancypagestyle{nosection}{%
\fancyhead[CE]{}
\fancyhead[CO]{}
\fancyfoot[CE,CO]{\thepage}
}

\DeclareMathOperator{\R}{\mathbb{R}}
\DeclareMathOperator{\C}{\mathbb{C}}

\renewcommand{\tilde}{\widetilde}
\renewcommand{\hat}{\widehat}

\renewcommand{\qedsymbol}{$\blacksquare$}

\usepackage{fancyhdr}

\fancypagestyle{plain}{%
\fancyhead[C]{}
\fancyfoot[C]{}

}
\fancypagestyle{nosection}{%
\fancyhead[CE]{}
\fancyhead[CO]{}
\fancyfoot[CE,CO]{\thepage}
}

\DeclareMathOperator{\BO}{\mathcal{O}}

\DeclareMathOperator{\capmat}{\mathcal{C}}

\DeclareMathOperator{\capmatg}{\mathcal{C}^\gamma}

\renewcommand{\epsilon}{\varepsilon}
\DeclareMathOperator{\dd}{d\!}

\renewcommand{\tilde}{\widetilde}
\renewcommand{\hat}{\widehat}

\DeclareMathOperator{\iL}{{\mathsf{L}}}
\DeclareMathOperator{\iR}{{\mathsf{R}}}
\DeclareMathOperator{\iLR}{{\mathsf{L},\mathsf{R}}}

\usepackage{tcolorbox}
\usetikzlibrary{tikzmark,calc}

\usepackage{enumerate}
\usepackage{upgreek}